\documentclass[smallextended]{svjour3} 

\usepackage{mathptmx}       
\usepackage{amsmath} 
\usepackage{helvet}         
\usepackage{courier}        
\usepackage{type1cm}        
 \usepackage[hyphens]{url}             
\usepackage{amssymb}
\usepackage{makeidx}         
\usepackage{array,multirow,graphicx}      
\usepackage{multicol}        
\usepackage{tikz}
\usetikzlibrary{arrows,positioning}
\tikzset{
    >=stealth,
    font=\scriptsize,
    possible world/.style={circle,draw,thick,align=center},
    real world/.style={double,circle,draw,thick,align=center},
    minimum size=40pt
}
\tikzstyle{vertex}=[circle, draw, inner sep=0pt, minimum size=6pt]

\usepackage[bottom]{footmisc}

\usepackage{makecell}

\renewcommand{\H}{\mathcal{H}}
\makeindex             

\begin{document}

\title{A Centrality Measure for Cycles and Subgraphs II} 
\titlerunning{A Centrality Measure for Cycles and Subgraphs II} 

\author{Pierre-Louis Giscard and Richard C. Wilson}
\institute{P.-L. Giscard \at
              University of York, Department of Computer Science, Deramore Lane, Heslington, York, YO10 5GH, United Kingdom. \\
              \email{pierre-louis.giscard@york.ac.uk}           
           \and
            R. C. Wilson \at
              University of York, Department of Computer Science, Deramore Lane, Heslington, York, YO10 5GH, United Kingdom.
}
\authorrunning{P.-L. Giscard and R. C. Wilson}
\date{\today}
\institute{Pierre-Louis Giscard \at Department of Computer Science, University of York, Deramore Lane, Heslington, York, YO10 5GH, United Kingdom, \email{pierre-louis.giscard@york.ac.uk}
\and Richard C. Wilson \at Department of Computer Science, University of York, Deramore Lane, Heslington, York, YO10 5GH, United Kingdom, \email{richard.wilson@york.ac.uk}}
%
%
\maketitle

\abstract{In a recent work we introduced a measure of importance for groups of vertices in a complex network. This centrality for groups is always between 0 and 1 and induces the eigenvector centrality over vertices. Furthermore, its value over any group is the fraction of all network flows intercepted by this group. Here we provide the rigorous mathematical constructions underpinning these results via a semi-commutative extension of a number theoretic sieve. We then established further relations between the eigenvector centrality and the centrality proposed here, showing that the latter is a proper extension of the former to groups of nodes. We finish by comparing the centrality proposed here with the notion of group-centrality introduced by Everett and Borgatti on two real-world networks: the Wolfe's dataset and the protein-protein interaction network of the yeast \textit{Saccharomyces cerevisiae}. In this latter case, we demonstrate that the centrality is able to distinguish protein complexes.} 

\section{Introduction}
\subsection{Context}
In our previous work on the subject, we argued the need to go beyond vertices when analysing complex networks. In fact, remarks to this end can be found scattered in the literature \cite{Contreras2014,Estrada2005,Milo2002,Mukhtar2011,Yeger2004}. For example, studies of gene regulatory networks have shown that ``motif-based centralities outperforms other methods" and can discern interesting network features not grasped by more traditional vertex centralities \cite{Koschutzki2007,Koschutzki2008}. Another example is provided by the notion of protein essentiality, a property now understood to be determined at the level of protein complexes, that is groups of proteins in the protein-protein interaction network (PPI) rather than at the level of individual proteins \cite{Hart2007,Ryan2013}. In addition, further biological properties have been tied to ensembles of genes or proteins, e.g. the notion of synthetic lethality, where the simultaneous deactivation of two genes is lethal while the separate deactivation of each is not \cite{Nijman2011}.
Since measures of importance for nodes constitute a key tool in the study of complex networks, it is only logical to expect that similar tools for ranking groups of vertices could find widespread applications throughout network analysis.\\[-.7em] 

In this spirit, we proposed in \cite{Giscard2017} a measure of importance for groups of nodes (henceforth called ``subgraphs"), that has the following desirable properties:
\begin{enumerate}
\item Provided the edge weights are non-negative, the centrality $c(H)$ of a subgraph $H$ is always between 0 and 1.
\item The precise value $c(H)$ taken by the centrality on a subgraph $H$ is the fraction of all network flows intercepted by $H$.
\item For subgraphs comprising a single node $H\equiv \{i\}$, the centrality measure $c(\{i\})$ yields the same ranking than the eigenvector centrality. In other terms, it induces the eigenvector centrality over vertices.
\item Computationally, $c(H)$ costs no more to compute per subgraph $H$ than ordinary vertex-centralities. What is computationally costly however, is to compute it over all subgraphs.  
\end{enumerate}
In \cite{Giscard2017}, we have shown, by analysing real-world networks from econometry and biology, that $c(.)$ performs better than centralities defined from naive sums of vertex-centralities. Concretely, we demonstrated that subgraph centralities defined from sums of the resolvent, exponential and eigenvector centralities failed to account for even the dominant events affecting input-output economic networks. In biology, we used $c(.)$ to construct a model of protein-targeting by pathogens that achieved a $25\%$ improvement over the state of the art one.\footnote{We refer to the area under the ROC curves for both the model based on the centrality $c(.)$ and the state of the art one. These are 0.97 and 0.73 respectively.}\\[-.5em] 

In this work, we establish further properties of the centrality measure $c(.)$ and present its rigorous mathematical underpinnings. We also compare this centrality with the notion of group-centrality presented by Everett and Borgatti in \cite{Everett1999} on real-world networks. 

\subsection{Notations and centrality definition}\label{Notation1}
The measure of cycle and subgraph centrality we propose is rooted in recent advances in the algebraic combinatorics of walks on graphs. Here we only define the few concepts from this background that are necessary to comprehend the centrality measure. 

We consider a finite network $G = (\mathcal{V} ;\mathcal{E})$ with $N=|\mathcal{V}|$ nodes and $M=|\mathcal{E}|$ edges and which may be weighted and directed. The adjacency matrix of $G$ is denoted $\mathsf{A}_G$ or simply $\mathsf{A}$. If $G$ is weighted then the entry $\mathsf{A}_{ij}$ is the weight of the edge $e_{ij}$ from $i$ to $j$ if this edge exists, and 0 otherwise.\\[-.85em]

A \textit{induced subgraph} $H$ of $G$, also called simply a \textit{subgraph} of $G$ and denoted $H\prec G$, is a set of vertices $ \mathcal{V}_H\subseteq \mathcal{V}$ together with the set of all edges linking these vertices in $G$, $\mathcal{E}_H=\{e_{ij}\in\mathcal{E}:\,i,j\in\mathcal{V}_H\}$.\\[-.85em]  

A \textit{walk} $w$ of length $\ell(w)$ from $v_i$ to $v_j$ on $G$ is a sequence $w = e_{i i_1} e_{i_1 i_2} \cdots e_{i_{\ell-1} j}$ of $\ell$ contiguous edges. The walk $w$ is \textit{open} if $i \neq j$ and \textit{closed} otherwise.\\[-.85em]

A \textit{simple cycle}, also known in the literature under the names \textit{loop}, \textit{cycle}, \textit{elementary circuit} and \textit{self-avoiding polygon}, is a closed walk $w = e_{i i_1} e_{i_1 i_2} \cdots e_{i_{\ell-1} i}$ which does not cross the same vertex twice, that is, the indices $i,i_1,\hdots,i_{\ell-1}$ are all different.\\

We now recall the definition of the centrality for cycles and subgraphs, introduced in \cite{Giscard2017}.
\begin{definition}[Centrality]
Let $G$ be a possibly weighted (di)graph, and let $\lambda$ be its maximum eigenvalue. Let $\mathsf{A}$ be the adjacency matrix of $G$, including weights if any. For any cycle $\gamma$, let $\mathsf{A}_{G\backslash \gamma}$ be the adjacency matrix of the graph $G$ where all vertices visited by $\gamma$ and the edges adjacent to them have been removed. Then we define the centrality $c(\gamma)$ of the cycle $\gamma$ as
\vspace{-2mm}
$$
c(\gamma):=\det\left(\mathsf{I}-\frac{1}{\lambda}\mathsf{A}_{G\backslash \gamma}\right).
$$
More generally, for any non-empty subgraph $H$ of $G$, we define the centrality of $H$ as 
$$
c(H):=\det\left(\mathsf{I}-\frac{1}{\lambda}\mathsf{A}_{G\backslash H}\right).
$$
\end{definition}

As we have shown in \cite{Giscard2017}, these centralities not only reflect the relative importance of cycles or subgraphs, but their values have a precise meaning too. Indeed, $c(H)$ is the fraction of all  information flows on the network that are intercepted by the subgraph $H$. As such, and as long as the network has no negative edge-weights, the centrality is always between 0 and 1, which is numerically advantageous,
$$
0\leq c(H)\leq 1.
$$
Because it has a concrete real-world meaning as fraction of network flows, the value of the centrality can be assessed with respect to external informations when available. More generally, it enriches the analysis in that it does not only produce a ranking of groups of nodes, but it also quantitatively ties these groups' importance with an immediately meaningful quantity, e.g. a fraction of capital flow, of successions of proteins interactions or of social interactions depending on the context.\\[-.5em] 

It the following section we give the full, rigorous mathematical proof of the main theorem underpinning these results and which relates the centrality $c(\gamma)$ of a cycle $\gamma$ with network flows. This theorem was presented as Proposition~1 in \cite{Giscard2017} but was only given a qualitative proof there, owing to length constraints. Note, we focus on the centrality of simple cycles as it is precisely in this context that the rigorous proof appears as an extension of a number theoretic sieve. The case of arbitrary subgraphs is similar, and we operate with no loss of generality.

\section{Centrality and network flows: a rigorous mathematical proof}
We first need to recall some combinatorial notions introduced in the context of the extension of number theory satisfied by walks on graphs \cite{Giscard2016}. The central objects of this earlier study are \emph{hikes}, a hike $h$ being an unordered collection of disjoint closed walks. Hikes can be also be seen as equivalence classes on words $W=\gamma_{i_1}\gamma_{i_2}\cdots \gamma_{i_n}$ over the alphabet of simple cycles $\gamma_i$ of a graph. Two words $W$ and $W'$ are equivalent if and only if $W'$ can be obtained from $W$ through allowed permutations of consecutive simple cycles. In this context, two simple cycles are allowed to commute if and only if they are vertex disjoint $\mathcal{V}(\gamma_i)\cap \mathcal{V}(\gamma_j)=\emptyset \iff \gamma_i\gamma_j=\gamma_j\gamma_i$.\\[-.5em] 

For example, if $\gamma_1$ and $\gamma_2$ commute but neither commute with $\gamma_3$, then $\gamma_1\gamma_2$ and $\gamma_2\gamma_1$ represent the same hike, but $\gamma_1\gamma_3\gamma_2$ and $\gamma_2\gamma_3\gamma_1$ are distinct hikes.\\[-.5em]

The letters $\gamma_{i_1},\cdots, \gamma_{i_n}$ found in a hike $h$ are called its prime divisors. This terminology is due to the observation that simple cycles obey the defining property of prime elements in the semi-commutative monoid $\mathcal{H}$ of hikes. In addition, they constitute the formal, semi-commutative, extension of prime numbers \cite{Giscard2016}.\\[-.7em]
 
Two special types of hikes will be important for our purpose here:\\
A \emph{self-avoiding hike} is a hike all prime factors of which commute with one another. In other terms, it is collection of vertex-disjoint simple cycles.

A walk, defined earlier in section~\ref{Notation1}, can be shown to be hikes with a unique right prime divisor \cite{Giscard2016}, a characterisation which is both necessary and sufficient so that any hike with a unique right prime divisor is a walk.

It may perhaps help the reader's intuition to know that in the extension of number theory satisfied by hikes, hikes are the extension of the integers, self-avoiding hikes are the square-free integers and walks are integers of the form $p^k$, with $p$ prime and $k\in\mathbb{N}$. \\[-.7em]

Now we claim that the centrality $c(\gamma)$ of a simple cycle $\gamma$ is exactly the fraction of all hikes $h$ (including infinite length ones) such that all right prime divisors of $h$ intercept $\gamma$, that is no right prime divisor of $h$ is vertex-disjoint with $\gamma$ and commutes with it. This later observation implies that $\gamma$ is the only right prime divisor of $h\gamma$. Thus, the claim we make is equivalent to stating that $c(\gamma)$ is the proportion of all hikes $h$ such that $h\gamma$ is a walk.
\begin{theorem}
Let $G$ be a finite (di)graph with adjacency matrix $\mathsf{A}$ and let $\gamma$ be a simple cycle on $G$. Then the total number $n_\gamma(k)$ of closed walks of length $k$ on $G$ with right prime divisor $\gamma$ is asymptotically equal to
\vspace{-2mm}
$$
n_\gamma(k)\sim c(\gamma)\!\left(\frac{1}{\det(\mathsf{I}-z\mathsf{A})}\right)\![k],~~\text{as}~~k\to\infty,
\vspace{-2mm}
$$
where $\big(1/\det(\mathsf{I}-z\mathsf{A})\big)[k]$ stands for the coefficient of order $k$ in the series $1/\det(\mathsf{I}-z\mathsf{A})$.
\end{theorem}
\begin{proof}
The proof relies on a very general combinatorial sieve. 
Let $\mathcal{H}_\ell:=\{h\in\H:~\ell(h)=\ell\}$ be the set of hikes of length $\ell$, $\mathcal{P}\subsetneq\H$ be a set of primes and $\mathcal{P}^{\text{s.a.}}$ the set of all self-avoiding hikes constructible from $\mathcal{P}$. Let $S(\mathcal{H}_\ell,\mathcal{P})$ be the number of hikes in $\mathcal{H}_\ell$ which are not right-divisible by any prime of $\mathcal{P}$.
The semi-commutative extension of the sieve of Erathostenes-Legendre yields
$$
S(\mathcal{H}_\ell,\mathcal{P}) = \sum_{d\in \mathcal{P}^{\text{s.a.}}} \mu(d) |\mathcal{M}_d|,
$$
with $|\mathcal{M}_d|$ the number of multiples of $d$ in $\mathcal{H}_\ell$. Furthermore, $\mu(d)$ is the M\"{o}bius function on hikes, which is \cite{Giscard2016}
$$
\mu(h)=\begin{cases}
(-1)^{\Omega(h)},&\text{if $h$ is self-avoiding,}\\
0,&\text{otherwise,}
\end{cases}
$$  
where $\Omega(h)$ is the number of prime divisors of $h$, including multiplicity.\\[-.7em] 

In order to progress, we seek a multiplicative function $\text{prob}(.)$ such that $|\mathcal{M}_d| = \text{prob}(d)|\mathcal{H}_\ell| + r(d)$, $|\mathcal{H}_\ell|:=\text{card}(\mathcal{H}_\ell)$. In this expression, $\text{prob}(d)$ approximates the probability that a hike taken uniformly at random in $\mathcal{H}_\ell$ is right-divisible by $d$. If edge-weights are present, the hikes are not all uniformly probable but follow a distribution dependent on these weights. In any case, no knowledge of this distribution is required here and the meaning of $\text{prob}(.)$ is only mentioned to help the reader understanding.  Similarly, $m(d)=\text{prob}(d)|\mathcal{H}_\ell|$ is the expected number of multiples of $d$ in $\mathcal{H}_\ell$. Finally, $r(d)$ is the associated error term, arising from the fact that $|\mathcal{M}_d|$ is not truly multiplicative. 
Supposing that we can identify the $m(.)$ function, we would obtain 
 $$
S(\mathcal{H}_\ell,\mathcal{P}) = \sum_{d\in \mathcal{P}^{\text{s.a.}}} \mu(d) m(d)  + \sum_{d\in \mathcal{P}^{\text{s.a.}}} \mu(d) r(d).
$$
Contrary to number theory, the first term does not admit any simpler form without further assumptions on  $\mathcal{P}$. This is because of the possible lack of commutativity between some elements of $\mathcal{P}$. We note however that since $\mu(d)$ is non-zero if and only if $d$ is self-avoiding, and since we have required that $m(.)$ be multiplicative,\footnote{But not necessarily totally multiplicative.} then it follows that the first term is entirely determined from the values of $m(.)$ over the primes of $\mathcal{P}$.\\

We therefore turn to determining $m(\gamma)$ for $\gamma$ prime. 
The set of left-multiples of $\gamma$ in $\H$ is $\{h\gamma,~h\in \H\}$, which is in bijection with the set $\{h\in \H,~\ell(h)\geq \ell(\gamma)\}$. Thus, the number of left-multiples of $\gamma$ in $\mathcal{H}_{\ell}$, is exactly $|\mathcal{H}_{\ell-\ell(\gamma)}|$. 
Then
$$
\text{prob}(\gamma) + \frac{r(\gamma)}{|\mathcal{H}_{\ell}|} = \frac{|\mathcal{H}_{\ell-\ell(\gamma)}|}{|\mathcal{H}_{\ell}|}.
$$
Seeking the best possible probability function $\text{prob}(\gamma)$, let us suppose that once this function has been chosen, the error term of the above equation vanishes in the limit $\ell\to\infty$. If this is true, then we obtain 
$$
\text{prob}(\gamma) = \lim_{\ell\to \infty}\frac{|\mathcal{H}_{\ell-\ell(\gamma)}|}{|\mathcal{H}_{\ell}|}.
$$
In order to progress, we make an important observation regarding the cardinality of the set $\mathcal{H}_{\ell}$:
\begin{lemma}\label{assumptionAG}
Let $G$ be a finite (directed) graph. Let $\mathcal{H}_{\ell}:=\{h\in\mathcal{H}:~\ell(h) =  \ell\}$ be set of all hikes on $G$ of length $\ell$. Then, there exists $\Lambda\in \mathbb{R}^+$ and a bounded function $f:\mathbb{N}\mapsto \mathbb{R}$ such that $\lim_{\ell\to\infty}f(\ell)$ exists and for $\ell\in\mathbb{N}^*$ we have exactly
\begin{equation*}
|\mathcal{H}_{\ell}|=\Lambda^\ell f(\ell).
\end{equation*}
If the absolute value of the largest eigenvalue $\lambda$ of $G$ has multiplicity $g$, then $\Lambda = \lambda^g$.
\end{lemma}
\begin{proof}
This follows directly from the ordinary zeta function on hikes $\zeta(z)=\det(\mathsf{I}-z\mathsf{A})^{-1}$, from which we have
\begin{equation*}
|\mathcal{H}_{\ell}|=\left(\frac{1}{\det(\mathsf{I}-z\mathsf{A})}\right)\![\ell]
=\sum_{i_1,\cdots,\, i_N\vdash \ell} \lambda^{i_1}_{1}\lambda^{i_2}_{2}\cdots \lambda^{i_N}_{N}=\lambda^\ell\!\! \sum_{i_1,\cdots,\, i_N\vdash \ell} \lambda^{i_1-\ell}\lambda^{i_2}_{2}\cdots \lambda^{i_N}_{N}
\end{equation*}
where the sums run over all positive values of $i_j \geq 0$ such that $\sum_j i_j = \ell$ and $\lambda\equiv \lambda_1$ is the eigenvalue of the graph with the largest absolute value. We assume for the moment that $\lambda$ is unique and let 
$
f(\ell):= \sum_{i_1,\cdots, \,i_N\vdash \ell} \lambda^{i_1-\ell}\lambda^{i_2}_{2}\cdots \lambda^{i_N}_{N}. 
$
This function is clearly bounded and 
$$
\lim_{\ell\to\infty} f(\ell) = \lim_{z\to\lambda^{-1}}(1-z\lambda)\zeta(z),
$$ 
exists and is finite. If $|\lambda|$ is not unique and has multiplicity $g$, then one should pick $\lambda^g$ for the scaling constant together with $f(\ell)=\zeta(z)[\ell]\lambda^{-g\ell}$. In all cases the Lemma follows. \end{proof}

Proceeding with the result of Lemma~\ref{assumptionAG} and assuming that the largest eigenvalue is unique for simplicity, the existence of the limit for $f$ gives 
$$
\text{prob}(\gamma) = \lim_{\ell\to \infty}\frac{\lambda^{\ell-\ell(\gamma)} f\big(\ell-\ell(\gamma)\big)}{\lambda^\ell f(\ell)} = \lambda^{-\ell(\gamma)}.
$$
The prob(.) function is multiplicative over the primes--recall these are the simple cycles--as desired. It yields $m(\gamma)  = |\mathcal{H}_{\ell}| \lambda^{-\ell(\gamma)}$ and the associated error term is 
\begin{align*}
r(\gamma) = |\mathcal{H}_{\ell-\ell(\gamma)}| -  |\mathcal{H}_{\ell}| \lambda^{-\ell(\gamma)} &= \lambda^{\ell-\ell(\gamma)}\Big(f\big(\ell-\ell(\gamma)\big)-f(\ell)\Big).
\end{align*}
%
To establish the validity of these results, we need only verify that they are consistent with our initial assumption concerning the error term, namely that $r(\gamma)/|\mathcal{H}_{\ell}|$ vanishes in the limit $\ell\to\infty$. The existence of the limit of $f$ implies $\lim_{\ell\to\infty}|f\big(\ell-\ell(\gamma)\big)-f(\ell)|=0$ and therefore that 
$$
\lim_{\ell\to\infty }\frac{r(\gamma)}{ |\mathcal{H}_{\ell}|}=\lim_{\ell\to\infty }\,\lambda^{-\ell(\gamma)}\Big(f\big(\ell-\ell(\gamma)\big)-f(\ell)\Big) =0,
$$
as required.\\

We are now ready to proceed with general self-avoiding hikes. Let $d=\gamma_1\cdots \gamma_q$ be self-avoiding. Then since $m$ is multiplicative and the length is totally additive over $\mathcal{H}$, $m(d) = \prod_i m(\gamma_i) = \lambda^{-\sum_i \ell(\gamma_i)} = \lambda^{-\ell(d)}$. The associated error term follows as
$$
r(d) = |\mathcal{H}_{\ell-\ell(d)}| -  |\mathcal{H}_{\ell}| \lambda^{-\ell(d)} = \lambda^{\ell-\ell(d)}\Big(f\big(\ell-\ell(d)\big)-f(\ell)\Big).
$$ 
Inserting these forms for $m(d)$ and $r(d)$ in the semi-commutative Erathostenes-Legendre sieve yields the sieve formula
 \begin{equation*}
S(\mathcal{H}_{\ell},\mathcal{P}) = |\mathcal{H}_{\ell}|\sum_{d\in \mathcal{P}^{\text{s.a.}}} \mu(d) \lambda^{-\ell(d)}  + \lambda^{\ell}\sum_{d\in \mathcal{P}^{\text{s.a.}}} \mu(d) \lambda^{-\ell(d)} \big(f(\ell-\ell(d))-f(\ell)\big).
\end{equation*}
We can now progress much further on making an additional assumption concerning the nature of the prime set $\mathcal{P}$. We could consider two possibilities: i) that  $\mathcal{P}$ is the set of all primes on an induced subgraph $H\prec G$; or ii) that $\mathcal{P}$ is a cut-off set, e.g. one disposes of all the primes of length $\ell(\gamma)\leq \Theta$. Remarkably, in number theory, if i) is true then ii) is true as well, and the sieve benefits from the advantages of both situations. In general however, i) and ii) are not compatible and while ii) could be used to obtain direct estimates for the number of primes of any length, a problem of great interest, we can show that this makes the sieve NP-hard to implement. We therefore focus on the first situation.\\

Let $H\prec G$ be an induced subgraph of the graph $G$ and let that $\mathcal{P}\equiv \mathcal{P}_H$ be the set of all primes (that is simple cycles) on $H$. Remark that $\sum_{d\in\mathcal{P}_H^{s.a}}\mu(d) \lambda^{-\ell(d)}$ is therefore the sum over all the self-avoiding hikes on $H$, each with coefficient $\mu(d)\lambda^{-\ell(d)}$. It follows \cite{Giscard2016} that $\sum_{d\in\mathcal{P}_H^{s.a}}\mu(d) \lambda^{-\ell(d)}=\det(\mathsf{I}-\lambda^{-1}\mathsf{A}_H)$. Concerning the error term,
$$
\lambda^{\ell}\sum_{d\in \mathcal{P}_H^{\text{s.a.}}} \mu(d) \lambda^{-\ell(d)} \big(f(\ell-\ell(d))-f(\ell)\big),
$$
we note that since $H$ is finite,\footnote{$G$ is finite and so are all its induced subgraphs.} the above sum involves finitely many self-avoiding hikes $d$. In addition,
given that $\lim_{\ell\to\infty}f(\ell)$ exists by Lemma~\ref{assumptionAG}, $\lim_{\ell\to\infty}f(\ell-\ell(d))-f(\ell) =0$ as long as $\ell(d)$ is finite, which is guaranteed by the finiteness of $H$. We have consequently established that the error term comprises finitely many terms, each of which vanishes in the $\ell\to\infty$ limit. As a corollary, the first term becomes asymptotically dominant: 
 \begin{equation*}
S(\mathcal{H}_{\ell},\mathcal{P}) \sim |\mathcal{H}_{\ell}|\det\big(\mathsf{I}-\lambda^{-1}\mathsf{A}_H\big) ~~\text{as}~~\ell\to\infty.
\end{equation*}
We can make this more explicit on using the ordinary form of the zeta function on hikes $\zeta(z)=1/\det(\mathsf{I}-z\mathsf{A})$. Then $|\mathcal{H}_{\ell}|=\zeta(z)[\ell]$ is the coefficient of order $\ell$ in $\zeta(z)$, see also the proof of Lemma~\ref{assumptionAG}.\\[-.5em]

\begin{remark}
The error term can be given a determinantal form upon using a finite difference expansion of $f$ or a Taylor series expansion of it if one smoothly extends its domain from $\mathbb{N}$ to $\mathbb{R}$. Writing
$$
f(\ell-\ell(d))-f(\ell) = \sum_{k\geq 1} \frac{\nabla^k[f](\ell)}{k!}\big(\ell(d)\big)_{(k)},
$$
with $(a)_{(k)}=\prod_{i=0}^{k-1}(a-i)$ the falling factorial and $\nabla$ the backward difference operator. Now we use the properties of the M\"{o}bius function on hikes to write $\sum_{d\in\mathcal{P}_H^{s.a}}\mu(d)\big(\ell(d)\big)_{(k)} \,z^{\ell(d)}= (\frac{d}{dz})^k \det(\mathsf{I}-z\mathsf{A}_H)$ and finally
\begin{align*}
S(\mathcal{H}_{\ell},\mathcal{P}) &= |\mathcal{H}_{\ell}|\,\det\!\left(\mathsf{I}-\frac{1}{\lambda}\mathsf{A}_H\right)+\lambda^{\ell}\sum_{k\geq 1}\frac{\nabla^k[f](\ell)}{k!} \det\!^{(k)}\!\Big(\mathsf{I}-\frac{1}{\lambda}\mathsf{A}_H\Big).\\[-1em]
\end{align*}
Here $\det\!^{(k)}\!\big(\mathsf{I}-\frac{1}{\lambda}\mathsf{A}_H\big)$ is a short-hand notation for $\big\{(\frac{d}{dz})^k \det(\mathsf{I}-z\mathsf{A}_H)\big\}\Big|_{z=\lambda^{-1}}$.\\[-.1em]
\end{remark}

To conclude the proof of the Theorem, 
we now need only choose $H$ correctly. Recall that we seek to count those walks which are left-multiples of a chosen simple cycle $\gamma$. But for $w=h\gamma$ to be a walk, the hike $h$ must be such that none of its right-prime divisor commutes with $\gamma$. This way, $\gamma$ is guaranteed to be the unique prime that can be put to the right of $h$, hence the unique right-prime divisor of $w$, making $w$ a walk. Then the sieve must eliminate all hikes $h$ with are left-multiples of primes \emph{commuting} with $\gamma$. Observe that all such primes are on $H=G\backslash \gamma$.\qed\\[-.5em]

%
\end{proof}

\begin{remark}
The construction presented here is much more general than appears at first glance. In particular, it can be extended to any additive function $\rho:\,\mathcal{H}\mapsto \mathbb{R}$ over $\mathcal{H}$ other than the length, provided an equivalent of Lemma~\ref{assumptionAG} exists for $\rho$. Infinite graphs may also be considered, provided additional constraints on the notion of determinant are met. These generalisations have further applications which will be presented elsewhere.  
\end{remark}

\section{Comparison with Everett and Borgatti's group-centralities}
\subsection{Motivations and context}
In our previous work on the centrality $c(H)$ \cite{Giscard2017}, we have presented comparisons with centralities obtained for $H$ upon summing up the vertex centralities of individual vertices involved in $H$. We have shown the comparative failure of these approaches which could not, for example, detect even the major crisis affecting the insurance$-$finance$-$real-estate triad in input-output networks over the period 2000-2014 period.\\[-.7em] 

In this section, we propose to further compare $c(H)$ with the notion of group centrality as it was introduced by Everett and Borgatti in 1999 \cite{Everett1999}. The authors of this study proposed to extend any vertex centrality to groups of vertices by summing up the centrality of the vertices of the group as calculated on a graph where other members of the group have been deleted. For example, the degree group centrality of an ensemble $H$ of vertices is equal to the external degree of $H$ in $G$.
Essentially, this approach is expected to characterise the importance of the group with respect to the rest of the graph but will not be sensitive to the inner structure of the group. As a consequence, it is easy to construct synthetic graphs where group-centralities 'fail' to identify a group that should clearly be the most central. For example, a sparse graph with a single large clique can be built such that this clique is less central than a small outlier group of nodes. In our opinion however, these limitations are more theoretical than practical and it is much more important to study the behaviour of the proposed measures on \emph{real-world} networks.

\subsection{The centrality $c(.)$ as an extension of the eigenvector centrality}
Incidentally, Everett and Borgatti provide a strong motivation for the development of a centrality akin to the one we propose here. Indeed, noting the lack of extension for the eigenvector centrality to groups of nodes in their work, they explain that ``[The eigenvector centrality] is virtually impossible to generalise along the lines presented earlier", that is, lest one resorts to node-merging, a procedure not without problems \cite{Everett1999}. Now recall that the centrality presented here $c(H)$ induces the eigenvector centrality on singleton subgraphs comprising exactly one vertex $H=\{i\}$, a requirement which, following Everett and Borgatti, is sufficient to call $c(.)$ a proper extension of the eigenvector centrality to groups of nodes.
In fact, this observation is itself a special case of a more general construction relating the centrality of simple paths with entries of the projector onto the dominant eigenvector:
\begin{proposition}
Let $G$ be a finite undirected graph with $\{\lambda\equiv \lambda_1, \lambda_2,\cdots ,\lambda_N\}$ its spectrum. For simplicity, we assume that the largest eigenvalue $\lambda$ of $G$ is unique. Let $W:\mathcal{E}\mapsto \mathbb{R}^+$ be the weight function, sending edges of the graph to their weights. If $G$ is not weighted then $W$ is identically 1. 
Let $\mathsf{P}_\lambda$ be the projector onto the dominant eigenvector of $G$ and $\eta:=\prod_{i>1}^N(1-\lambda_i/\lambda)$. Then
$$
\eta(P_\lambda)_{ij} =\sum_{p:\,i\to j} \lambda^{-\ell(p)}W(p)\,c(p),
$$ 
where the sum runs over all simple paths from $i$ to $j$ and the weight of a path is the product of the weights of the edge it traverses.
\end{proposition}
\begin{remark}
 When $i\equiv j$, the only simple path from $i$ to itself is the length 0 path that is stationary on $i$. The weight of the empty path is the empty product with value 1 and therefore we recover the result of \cite{Giscard2017}
$$
\eta eig(i)^2=\eta (\mathsf{P}_\lambda)_{ii} = c(\{i\}),
$$ 
where $eig(i)$ is the $i$th entry of the dominant eigenvector.
\end{remark}
\begin{proof}
This relation follows from e.g. the path-sum formulation of the resolvent function $\mathsf{R}(z):=\big(\mathsf{I}-z\mathsf{A})^{-1}$ \cite{Giscard2013}. We have
$$
\mathsf{R}(z)_{ij} =\sum_{p:\,i\to j} z^{\ell(p)}W(p)\,\frac{\det(\mathsf{I}-z\mathsf{A}_{G\backslash p})}{\det(\mathsf{I}-z\mathsf{A})}.
$$
In particular, the case $i\equiv j$ gives the well-known adjugate formula for the inverse $\mathsf{R}(z)_{ii}=\det(\mathsf{I}-z\mathsf{A}_{ G \backslash i})/\det(\mathsf{I}-z\mathsf{A})$.
Introducing the adjugate matrix $\mathrm{Adj}(\mathsf{I}-z\mathsf{A})_{ij} := \det(\mathsf{I}-z\mathsf{A})\mathsf{R}(z)_{ij}$ explicitly we have 
$$
\mathrm{Adj}(\mathsf{I}-z\mathsf{A})_{ij} =\sum_{p:\,i\to j} z^{\ell(p)}W(p)\,\det(\mathsf{I}-z\mathsf{A}_{G\backslash p}),
$$
and the result follows on noting that $\lim_{z\to1/\lambda}\mathrm{Adj}(\mathsf{I}-z\mathsf{A}) = \eta \mathsf{P}_\lambda$.\\[-.5em]
\end{proof}

We can go further to establish the centrality $c(.)$ as an extension of the eigenvector centrality to groups of nodes along broadly similar lines as those advocated by Everett and Borgatti. To introduce the main result here, we need to present the (intuitive) definitions of union and intersection of subgraphs.

Let $H,\,H'$ be two subgraphs of $G$. We designate by $H\cup H'$ the subgraph of $G$ whose vertex set is the set-theoretic union of the vertex sets of $H$ and $H'$, $\mathcal{V}(H\cup H')=\mathcal{V}(H)\cup\mathcal{V}(H')$. Similarly $H\cap H'$ is the subgraph of $G$ with vertex set $\mathcal{V}(H)\cap\mathcal{V}(H')$. 

\begin{proposition}
Let $G$ be a finite graph with no negative weights and $\{H_1,\cdots H_n\}$ be a set of connected induced subgraphs of $G$. Then
$$
c\Big(\bigcup_{i=1}^nH_i\Big)=\sum_{S\subseteq\{1,\cdots,n\}}(-1)^{|S|-1} c\Big(\bigcap_{s\in S} H_s\Big) .
$$
\end{proposition}
\begin{proof}
This follows from the definition of $c(H)$ as the fraction of all network flows intercepted by $H$. A direct application of the inclusion-exclusion principle gives the result. 
\end{proof}

An immediate corollary then explicitly shows how the centrality $c(.)$ of any group of nodes arises from the interplay between their eigenvector centralities
\begin{corollary}
Let $G$ be a finite graph with no negative weights. Let $\mathcal{V}_H:=\{v_1,\cdots,v_n\}\subseteq\mathcal{V}$ be a group of nodes on $G$. Then
$$
c\big(\{v_1,\cdots,v_n\}\big)=\eta\,\sum_{i=1}^n eig(v_i)^2 - \sum_{i,j\in \mathcal{V}_H}f(\{v_i,v_j\})+ \sum_{i,j,k\in \mathcal{V}_H}f(\{v_i,v_j,v_k\})-\cdots,
$$
where $f(\{v_i,v_j, v_k,\cdots\})$ is the fraction of all network flows intercepted by all of $v_i$, $v_j$, $v_k$, etc.
\end{corollary}

%
%

\subsection{Wolfe's dataset}
\begin{table}
\centering
\textbf{Centralities of groups of monkeys in Wolfe's dataset\\}
\vspace{1mm}
\begin{tabular}{cccccc}
\textbf{Group} & \textbf{Members} & $\mathbf{c(H)}$\textbf{ in $\%$} & \thead{Degree\\group centrality}& \thead{Average closeness \\group centrality}&\thead{Group \\betweenness} \\
\hline
Age 10$-$13 & 2~3~8~12~16 & 67$\%$ &   11&15&43.5\\
Age 7$-$9 & 4~5~9~10~15~17 & 57$\%$ & 5 &13.7&0 \\
Age 14$-$16 & 1~6~11~13~19& 49$\%$ & 8 &18&2.84 \\
Age 4$-$6 &7~14~18~20&34$\%$ & 5 &20.5&0 \\[.5em]
Females &$6-20$&95$\%$ & 4&6.4&0.5  \\
Males &$1-5$&67$\%$ & 10&16&24.34 \\
\hline
\end{tabular}
\caption{\label{Monkeys} Comparison between several of Everett and Borgatti's group centralities \cite{Everett1999} and the centrality $c(H)$. 
 The centrality values for $c(H)$ are given here in $\%$ as they give the proportions of all successions of interactions between monkeys involving at least one member of the group. The centralities $c(H)$ was computed by the FlowFraction algorithm available on the Matlab File Exchange \cite{MatlabFiles}}
\vspace{-5mm}
\end{table}

We begin our concrete comparison with group-centralities on the Wolfe primate dataset \cite{UCINET}, a small real-world network which was studied by Everett and Borgatti. This dataset provides the number of times  monkeys of a group of 20 have been spotted together next to a river by the anthropologist Linda Wolfe. 

Our results are shown in Table.~\ref{Monkeys}. Here the properties that $c$ is always between 0 and 1 and that its values have actual meaning are clearly advantageous.
For example, we can now not only tell that the age group 10$-$13 is the most central, as  Everett and Borgatti noted, but we can concretely assert that $67\%$ of all flows of interactions between monkeys involved at least one member of this group. By flow (or chain) of interactions, we mean successions of interactions between monkeys, including interactions that may occur simultaneously. For example, we can have monkey 1 interact with 3, who then interacts with 8; while concurrently 2 meets 4 etc.\\[-.8em] 

Similarly, we note that almost $95\%$ of all flows of interactions involved at least one female, while this percentage dropped to $64\%$ for males, in spite of male 3 being the most central individual monkey in the entire group by all measures. Thus, according to $c(H)$ and contrary to all the group centralities reported here,\footnote{Everett and Borgatti also discuss normalisations of the group-centralities. In the case of the degree group-centrality, the normalisation is defined to be the degree group centrality divided by the number of nodes which do not belong to the group under consideration. Normalisations tends to rank females ahead of males as $c(H)$ does, but they represent non-linear transformation of the original group-centralities, making their interpretation more difficult.} females are quantitatively more important in mediating social interactions than the males. Here, it may help to know that the monkeys observed by Wolfe were feral Rhesus macaques (\textit{Macaca mulatta}), a species where females stay in the group of their birth, providing its dominance rank structure, while males must change group when reaching sexual maturity, around 4 years old. Furthermore, during the mating season, females favour multiples interactions with different males including low ranking ones \cite{Lindburg1971}. Finally, females typically outnumber males, sometimes by as much as 3 to 1. These observations suggest that females should indeed account for a larger share of the all interactions between monkeys than the males.\\[-.8em] 

Another point of importance for the comparison is the age group 7$-$9, which is ranked higher than the age group 14$-$16 by $c(H)$ while the group-centralities consistently yield the opposite order. On this point, we observe that Rhesus macaques are peculiar in that younger females have higher social ranks than their older peers \cite{Hill1996,Waal1993}. In the closely related Japanese macaques (\textit{Macaca fuscata}), dominance rank is known to be positively correlated with the frequency of social interactions \cite{Singh1992}.

\subsection{Yeast PPI network and protein complexes}
In this section we study the PPI network of the yeast \textit{Saccharomyces cerevisiae}, using high quality data from \cite{Hart2007}, which provides a network comprising 5303 interactions between 1689 individual proteins. These proteins are known to belong to complexes, a curated list of which is provided by the Munich Information center on Protein Sequences (MIPS) \cite{Guldener2006}. The authors of \cite{Hart2007} have shown that some of the MIPS complexes could be recovered from a run of the MCL clustering algorithm running on the network. Our goal here is twofold: i) to show that the centrality $c(.)$ can also be used to recover MIPS protein complexes, for which it provides additional informations; and ii) that the degree group centrality fails to do so. Here, we focus specifically on the degree group centrality as the degree centrality is the vertex measure of importance which has seen the most success in biology, see e.g. \cite{Mukhtar2011}.\\[-.5em]

\begin{figure}[t!]
\vspace{-2mm}
\centering
\includegraphics[width=.8\linewidth]{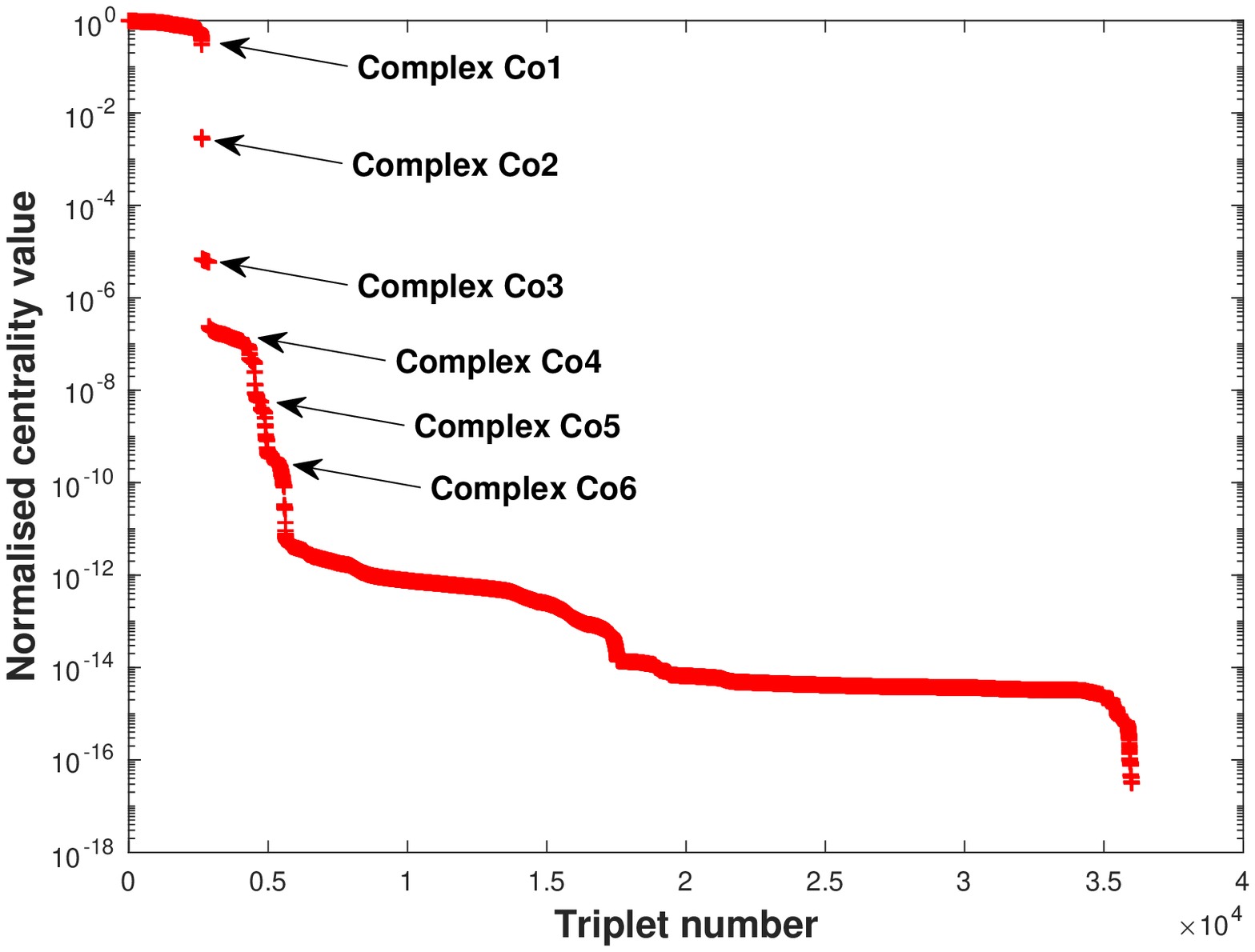}\\
\includegraphics[width=.8\linewidth]{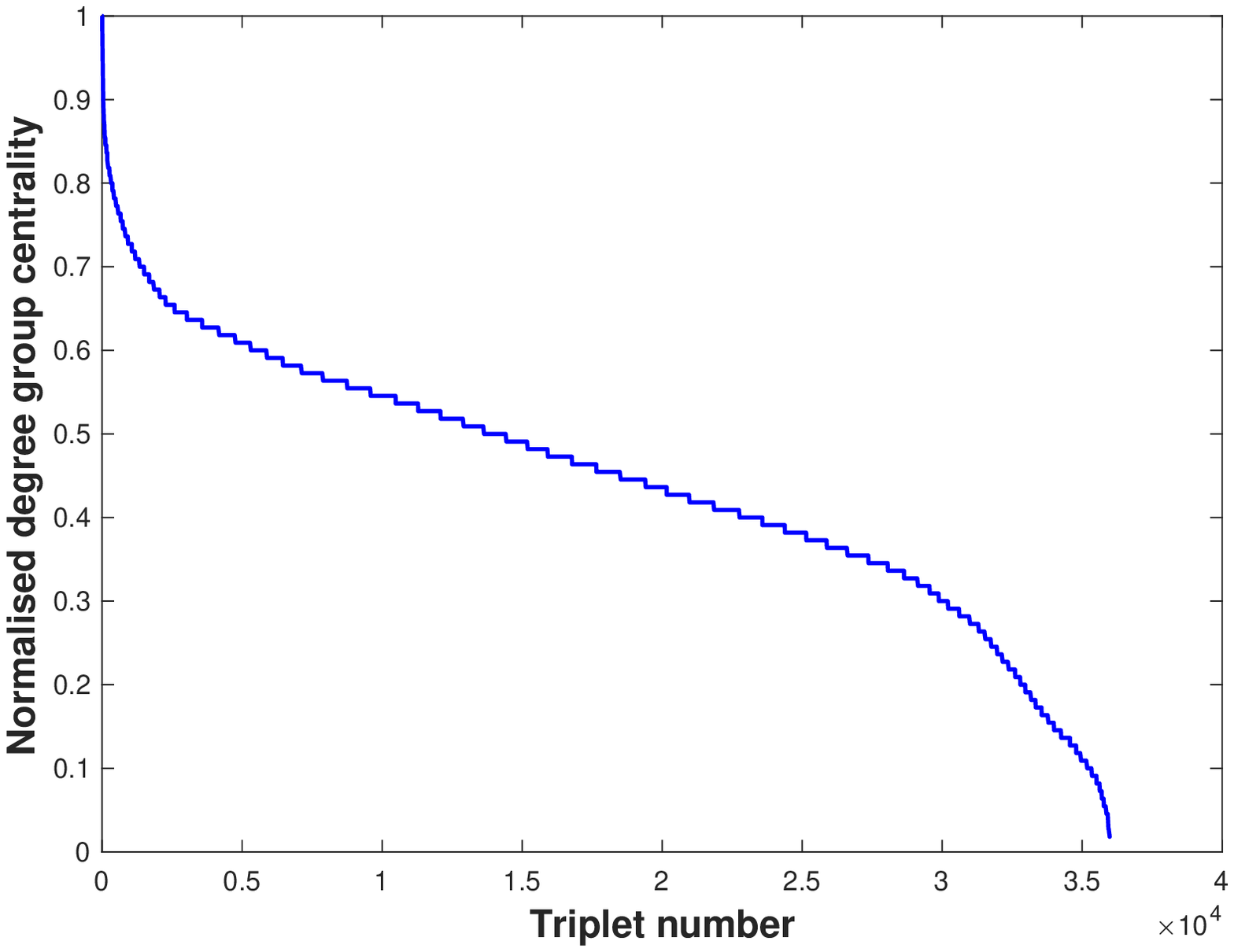}
\caption{\label{triplets} Distributions of triplet centralities. Top: normalised triplet centralities $c(t)/\max_{t\text{ triplet}}\big(c(t)\big)$, bottom: normalised degree group centrality $g(t)/\max_{t\text{ triplet}}\big(g(t)\big)$ introduced in \cite{Everett1999}.}
\end{figure} 
We calculated the centralities $c(.)$ and degree group centralities of all edges, connected triplets and connected quadruplets of proteins on the network. What is interesting here is the distribution of centrality values, which we show in Fig.~(\ref{triplets}) in the case of triplets.\footnote{Edges and quadruplets give broadly similar distributions. While complexes Co1, Co2 and Co3 are just as markedly visible in quadruplet data as in triplet data, quadruplets do lead to better segregation of complexes Co4, Co5 and Co6.} In the case of the centrality $c(.)$ proposed here, the distribution of triplet centrality values is organised into separate plateau-like clusters, which actually reveal the underlying protein complexes. Recovering the list of proteins involved in these clusters yields complexes which can be found in curated databases \cite{Pu2009}. Mathematically, the fact that complexes lead to clustered plateau-like centrality values for triplets means that the frequency with which proteins belonging to these complexes are involved in successions of proteins reactions depends first and foremost on the complexes themselves. In other terms, the frequency of protein activation is determined at the complex level.\\[-.5em] 

The dominant complex, here denoted Co1, comprises 30 proteins\footnote{It comprises proteins ASF1, EHD3, FYV4, MAM33, MRP1, MRP4, MRP10, MRP13, MRP21, MRP51, MRPS5, MRPS8, MRPS9, MRPS16, MRPS17, MRPS18, MRPS28, MRPS35, NAM9, PET123, RSM7, RSM10, RSM18, RSM19, RSM22, RSM23, RSM24, RSM25, RSM26 and RSM27.} and is found in both the MIPS database and in \cite{Pu2009}, where it is known as the mitochondrial small ribosomal large subunit. Interestingly, Co1 is identical with the third largest complex recovered by the MCL algorithm running on the same dataset \cite{Hart2007}, with the addition of the proteins ASF1 and MAM33, a nucleosome assembly factor and a protein of the mitochondrial matrix involved in oxidative phosphorylation, respectively. In the latter case, we note that several complexes involving the MAM33 and proteins of mitochondrial small ribosomal large subunit have been proposed in experimental studies \cite{MAM33}.
Complex Co2 comprises 21 proteins.\footnote{These are ASF1, CDC48, CKA1, HAT1, HAT2, HHF1, HHF2, HHT2, HIF1, HIR2, PDS5, POB3, PSE1, PSH1, RAD53, RTG2, RTT106, SPT16, YDL156W, YIL070C and YKU70.} It includes the entire complex C17 determined by the MCL method \cite{Hart2007}, together with 6 additional proteins all which are been proposed to form complexes (in particular the HIR and Rad53p-Asf1p complexes) with one or more proteins of C17 in separate studies \cite{Pu2009} as well as in the MIPS database. 
Complex Co3 comprises 64 proteins and overlaps significantly with the nucleosomal protein and CID 14 and complexes of \cite{Pu2009}, the latter of which includes the Casein kinase II, RNA polymerase II and Cdc73/Paf1 complexes.\footnote{This complex is ASF1, CDC34, CDC48, CDC53, CDC73, CDC9, CHD1, CKA1, CKA2, CKB1, CKB2, CTR9, DOA1, GRR1, HAT1, HAT2, HHF1, HHF2, HHT2, HIF1, HIR1, HIR2, HOT1, HPC2, HTA1, KAP114, LEO1, MET30, MKT1, MRF1, NAP1, NPL4, ORC2, ORC3, ORC4, ORC5, PAF1, PDS5, PEX19, POB3, POL12, PSE1, PSH1, RAD27, RAD53, RPS1B, RRP7, RTF1, RTG2, RTT101, RTT106, SHP1, SKP1, SPO12, SPT16, TOP1, UFD1, ULP1, UTP22, YDL156W, YDR049W, YGR017W, YKU70 and YKU80}\\[-.5em]

An advantage of the classification method employed here is that, contrary to MCL, it allows for overlapping complexes, i.e. proteins which functions in different complexes, as is expected biologically. At the same time, a drawback is that small centrality values are not segregated well enough to clearly distinguish clusters of values and hence complex boundaries. At least three more complexes Co4, Co5 and Co6 could possibly be distinguished, all of which can be found in MIPS database, however these are less clear cut than the first three complexes and so are left out from this work.
Empirically, we found that this problem could be somewhat reduced by looking at quadruplets, quintuplets etc., but this comes at a great computational cost given the number of such objects. A random sampling scheme may be able to bypass this difficulty.\\[-.5em]  

In comparison, the distribution of degree group centrality shows no trace of the underlying protein complexes and reveals little more than the simple distribution of vertex degrees. 
While we do not recommend the use of the centrality $c(.)$ as a clustering tool owing to its greater computational cost than algorithms such as MCL, we believe that its performance in this domain bears witness to the sensitivity of the proposed centrality to underlying network features. Conversely, the notion of group-centrality may be too coarse to perceived such features in the data, at least in the case of PPI.

\section{Conclusion}
In this second work on the centrality $c(.)$, we have rigorously established its meaning as a fraction of network flows intercepted by any chosen ensembles of nodes. The centrality $c(.)$ not only induces the eigenvector centrality on vertices, but it is a proper extension of it through an application of the inclusion-exclusion principle on network flows. Finally, we have shown on two real-world networks that the centrality $c(.)$ is more sensitive to critical network features than existing group-centralities. In particular, the centrality of triplets of proteins in the PPI network of the yeast was sufficient to distinguish protein complexes found in curated databases of experimental results. We recall that in our previous study \cite{Giscard2017}, the centrality $c(.)$ already produced the best available model for pathogen targeting in \textit{Arabidopsis thaliana}, yielding a $25\%$ improvement of the state-of-the-art model of \cite{Mukhtar2011}. We hope that these results will spur further research on the use of the centrality in biology.

\section*{Declarations}
\footnotesize{
\noindent \textbf{Availability of data and material} Raw data concerning Wolfe's dataset and the PPI of the yeast can be found in \cite{UCINET} and \cite{Hart2007}, respectively. The algorithms used to compute the centrality values are available online, on the Matlab File Exchange \cite{MatlabFiles}.\\[-.5em]

\noindent \textbf{Authors' contributions} P.-L. Giscard performed the research and both P.-L. Giscard and R. C. Wilson wrote the article.\\[-.5em]

\noindent \textbf{Competing interests} P.-L. Giscard and R. C. Wilson declare no financial and non-financial competing interests.\\[-.5em]

\noindent \textbf{Funding} P.-L. Giscard is grateful for the financial support from the Royal Commission for the Exhibition of 1851. The Royal Commission played no role in the present study and had no influence on the analysis of the data.}\\[-.5em]

\noindent \textbf{Acknowledgement} We thank Paul Rochet of the Laboratoire Jean-Leray, Nantes, France, for stimulating discussions.


\begin{thebibliography}{99.}%
\bibitem{Contreras2014}
Contreras MGA, Fagiolo G (2014) {Propagation of economic shocks in input-output
  networks: A cross-country analysis}.
\newblock {\em Phys. Rev. E} 90:062812.

\bibitem{Estrada2005}
Estrada E, Rodr\'iguez-Vel\'azquez JA (2005) {Subgraph centrality in complex
  networks}.
\newblock {\em Physical Review E} 71:056103.

\bibitem{Everett1999}
Everett M. G., Borgatti S. P. (1999) {The centrality of groups and classes}
\newblock {\em J. Math. Sociol.} 23(3):181--201
 
\bibitem{MatlabFiles}
Giscard P.-L., Wilson R. C. (2017) Algorithm to calculate the cycle-centrality of selected cycles or subgraphs: \url{https://mathworks.com/matlabcentral/fileexchange/64678}. Algorithm to calculate the centrality of all connected induced subgraphs of fixed size: \url{https://mathworks.com/matlabcentral/fileexchange/64677}.


\bibitem{Giscard2013}
Giscard P.-L., Thwaite S. J., Jaksch D. (2013) {Evaluating matrix functions by resummations on graphs: the method of path-sums}
\newblock {\em SIAM J. Mat. Anal. Appl.} 34(2): 445-469.

\bibitem{Giscard2016}
Giscard P.-L. , Rochet P. (2017) {Algebraic combinatorics on trace monoids: Extending
  number theory to walks on graphs.}
\newblock {\em SIAM J. Discrete Math.} 31(2):1428--1453.

\bibitem{Giscard2017}
Giscard P.-L., Wilson R. C. (2017) {Cycle-Centrality in Economic and Biological Networks}. 
\newblock In {\em Complex Networks \& Their Applications VI.} Cherifi C., Cherifi H., Karsai M., Musolesi M. (eds). Studies in Computational Intelligence, vol 689. Springer, pp. 14$-$28.

\bibitem{Guldener2006}
G\"{u}ldener U., M\"{u}nsterk\"{o}tter M., Oesterheld M., Pagel P., Ruepp A., Mewes H. W., St\"{u}mpflen V. (2006) {MPact: the MIPS protein interaction resource on yeast.}
\newblock {\em Nucleic Acids Res} 34(Database issue): D436-441.


\bibitem{Hart2007}
Hart G. T., Lee I., Marcotte E. R. (2007) {A high-accuracy consensus map of yeast protein complexes reveals modular nature of gene essentiality}.
\newblock {\em BMC Bioinformatics} vol 8, p. 236. 

\bibitem{Hill1996}
Hill D., Okayasu N. (1996) {Determinants of dominance among female macaques: nepotism, demography and danger}.
\newblock In {\em Evolution and Ecology of Macaque Societies.} Fa, J. and D. Lindburg (eds.). Cambridge: Cambridge University Press.

\bibitem{Koschutzki2007}
Kosch\"{u}tzki D., Schw\"{o}bbermeyer H., Schreiber F. (2007) {Ranking of network elements based on functional substructures}.
\newblock {\em J. Theor. Biol. } 248(3): 471-479.

\bibitem{Koschutzki2008}
Kosch\"{u}tzki D., Schreiber F. (2008) {Centrality Analysis Methods for Biological Networks and Their Application to Gene Regulatory Networks}.
\newblock {\em Gene Regul. Syst. Bio.} 2: 193?201.

\bibitem{Lindburg1971}
Lindburg D. G. (1971) {The rhesus monkey in north India : an ecological and behavioral study}. 
\newblock In {\em Primate behavior: developments in field and laboratory research, vol. 2}, Rosenblum L. A. editor. New York : Academic Press.

\bibitem{Milo2002}
Milo R, et~al. (2002) {Network motifs: simple building blocks of complex
  networks}.
\newblock {\em Science} 298(5594):824--827.

\bibitem{Mukhtar2011}
Mukhtar MS, et~al. (2011) {Independently Evolved Virulence Effectors Converge
  onto Hubs in a Plant Immune System Network}.
\newblock {\em Science} 333(6042):596--601.

\bibitem{Nijman2011}
Nijman S. M. B. (2011) {Synthetic lethality: General principles, utility and detection using genetic screens in human cells}.
\newblock {\em FEBS Lett.} 585(1): 1?6.

\bibitem{Pu2009}
Pu S., Wong J., Turner B., Cho E.,Wodak S. J. (2009) {Up-to-date catalogues of yeast protein complexes}
\newblock {\em Nucleic Acids Res.} 37(3): 825?831.

\bibitem{Ryan2013}
Ryan C. J., Krogan N. J., Cunningham P., Cagney G. (2013) {All or Nothing: Protein Complexes Flip Essentiality between Distantly Related Eukaryotes}
\newblock {\em Genome Biol. Evol.} 5(6): 1049?1059.

\bibitem{Singh1992}
Singh M.,  D'Souza L.,  Singh M. (1996) {Hierarchy, kinship and  social interaction among Japanese monkeys (Macaca fuscata)}.
\newblock {\em J. Biosci.} vol. 17, issue 1, pp. 15-27. 

\bibitem{UCINET}
UCINET IV Datasets \url{http://vlado.fmf.uni-lj.si/pub/networks/data/ucinet/ucidata.htm#wolf}, retrieved February 2018. 

\bibitem{Waal1993}
Wall F. (1993) {Codevelopment of dominance relations and affiliative bonds in rhesus monkeys}. 
\newblock In {\em Juvenile Primates: Life History, Development, and Behavior.}, Pereira, M., and L. Fairbanks (eds.). New York: Oxford Oxford University Press.

\bibitem{MAM33}
\textit{Yeast Resource Center}, Public Data Repository (2018) {Protein MAM33}
  (\url{http://www.yeastrc.org/pdr/viewProtein.do?id=531248}).
  
\bibitem{Yeger2004}
Yeger-Lotem E, et~al. (2004) {Network motifs in integrated cellular networks of
  transcription-regulation and protein-protein interaction}.
\newblock {\em Proc. Natl. Acad. Sci. U.S.A.} 101(16):5934--5939.

\end{thebibliography}
\end{document}